\documentclass[letterpaper, 10 pt, conference]{ieeeconf}  

\IEEEoverridecommandlockouts                              

\usepackage{graphics,color} 
\usepackage[english]{babel}
\usepackage{epsfig} 
\usepackage{amsmath,mathtools,booktabs,multirow,tabularx}

\newtheorem{theorem}{Theorem}
\usepackage{upgreek,bm}
\usepackage{cases}
\usepackage{marginnote}
\usepackage{cite}
\usepackage[dvipsnames]{xcolor}
\usepackage[mathscr]{euscript}
 \let\mathscr\relax
\usepackage[scr]{rsfso}

\newtheorem{prop}{Proposition}
\DeclareMathOperator*{\argminA}{arg\,min} 
\DeclareMathOperator*{\argmaxA}{arg\,max} 
\usepackage{amssymb}
\newcommand{\onenorm}[1]{\left\lvert#1\right\rvert}

\allowdisplaybreaks

\title{\LARGE \bf
On Affine Policies for Wasserstein Distributionally Robust Unit Commitment}

\author{
Youngchae Cho and Insoon Yang 
\thanks{
This work was supported in part by the National Research Foundation of Korea funded by the MSIT(2020R1C1C1009766), and Samsung Electronics.
}
\thanks{
The authors are with the Department of Electrical and Computer Engineering, ASRI, SEPRI, Seoul National University, Seoul, South Korea
        {\tt\small \{youngchaecho,insoonyang\}@snu.ac.kr}}
}

\begin{document}

\maketitle
\thispagestyle{empty}
\pagestyle{empty}{
}
\begin{abstract}
This paper proposes a unit commitment (UC) model based on data-driven Wasserstein distributionally robust optimization (WDRO) for power systems under uncertainty of renewable generation as well as its tractable exact reformulation. The proposed model is formulated as a WDRO problem using an affine policy, which nests an infinite-dimensional optimization problem and satisfies the non-anticipativity constraint. To reduce conservativeness, we develop a novel technique that defines a subset of the uncertainty set with a probabilistic guarantee. 
Subsequently, the proposed model is recast as a semi-infinite programming problem that can be efficiently solved using existing algorithms. Notably, the scale of this reformulation is invariant with sample size. As a result, samples can be easily incorporated without using a sophisticated decomposition algorithm regardless of their number. Numerical simulations on 6- and 24-bus test systems demonstrate the economic and computational efficiency of the proposed model. 
\end{abstract}

\section{Introduction}
The unit commitment (UC) problem is a fundamental planning problem for power systems, which yields an optimal commitment and dispatch schedule of conventional generators for a demand forecast. However, the increasing use of uncertain renewable energy generation (REG) has posed challenges in solving it. To address the uncertainty, various stochastic optimization methods have been studied \cite{zheng2015stochastic}. 

The two most frequent stochastic optimization 
methods for the UC problem under the uncertainty of REG are stochastic programming (SP) and robust optimization (RO). The objective of SP is to minimize the expected total operating cost with respect to the true distribution of the uncertainty \cite{takriti1996stochastic}. However, as the true distribution is unknown and often approximated by an empirical distribution comprising a finite number of samples, the solution may cause unexpectedly high operating costs. On the other hand, RO minimizes the worst-case total operating cost over a set of uncertain scenarios \cite{cho2019box}. However, as probabilistic features of the uncertainty are ignored, solutions may be overly conservative. 

Recently, distributionally robust optimization (DRO) has attracted researchers in the field of power system planning and control \cite{guo2018data,pourahmadi2019distributionally,li2019distributionally,yang2019data}. The main goal of DRO is to minimize the worst-case expected total operating cost over a family of distributions of uncertainty, called an {\it ambiguity set}. The ambiguity set is built around a nominal or reference distribution (mostly an empirical distribution, as is assumed hereafter) so that it is large enough to include the true distribution while small enough to avoid being too conservative. Utilizing samples in a way that compensates for their incompleteness, DRO can provide a more reliable and less conservative solution than SP and RO, respectively. 

Among various types of ambiguity sets that have been studied in the literature of DRO (see \cite{rahimian2019distributionally} and the reference therein), this paper focuses on {\it Wasserstein balls} as they offer tractable formulations and out-of-sample performance guarantees \cite{esfahani2018data,gao2016distributionally,zhao2018data}. However, Wasserstein DRO is not straightforwardly applied to the UC problem. This is because any resulting UC model, when reformulated in a tractable form, may well have a scalability issue regarding the number of samples that make up the reference distribution. Indeed, most of the existing Wasserstein distributionally robust UC (WDRUC) models have this issue; see, e.g.,  \cite{zheng2020data} and \cite{gamboa2021decomposition}. 

To the best of the authors' knowledge, \cite{zhu2019wasserstein} is the only research paper to develop a WDRUC model that has a tractable reformulation with a scale independent of sample size. With this aim, an affine policy is used in \cite{zhu2019wasserstein}. However, the affine policy in \cite{zhu2019wasserstein} violates the non-anticipativity constraint. This constraint is such that, with it, decisions under uncertainty have to be made based on only the past and current observation of the uncertainty, which is reasonable. 
Without the constraint, it is unrealistically assumed in \cite{zhu2019wasserstein} that the future realization of REG can be observed. Thus, the model in \cite{zhu2019wasserstein} cannot be used in practice.

In this paper, a WDRUC model is proposed with its tractable exact reformulation. First, the proposed model is formulated as a WDRO problem that nests an infinite-dimensional worst-case expectation problem, relying on an affine policy. This affine policy builds on the reserve-modeling method studied for the robust UC (RUC) model in \cite{cho2019box} and thus meets the non-anticipativity constraint unlike the one in \cite{zhu2019wasserstein}. Subsequently, the proposed model is recast as a semi-infinite programming problem that can be efficiently solved using existing algorithms. The scale of this reformulation is independent of sample size. Thus, unlike in \cite{zheng2020data} and \cite{gamboa2021decomposition}, a large number of samples can be easily incorporated without using sophisticated decomposition algorithms. Meanwhile, the proposed model builds the Wasserstein ball on only a subset of the uncertainty set for more economical decision making. The subset is defined using a novel technique, which is computationally efficient and preserves the sample size independence of the solution method for the proposed model. Numerical simulation results demonstrate the economic and computational efficiency of the proposed model and the solution method, respectively.

In summary, the main contributions of this paper are as follows:
\begin{itemize}
\item A WDRUC model is proposed using an affine policy that satisfies the non-anticipativity constraint.
\item A tractable exact reformulation of the proposed model is provided, with a scale independent of sample size. 
\item A technique is developed to define a subset of the uncertainty set for the Wasserstein ball, preserving 
the solution method’s independence of sample size.
\end{itemize}

The rest of this paper is organized as follows: Section \ref{sec:preliminaries} explains a deterministic UC (DUC) model without considering any uncertainty and reviews the RUC model in \cite{cho2019box}. Section \ref{sec:exactWDRUC} formulates an ``exact" WDRUC model, which does not rely on an affine policy. Section \ref{sec:WDRUCwith} introduces the proposed model with affine policies and provides the tractable reformulation. Section \ref{sec:simulations} discusses the numerical simulation results. Finally, Section \ref{sec:conclusions} concludes this paper. 

\section{Preliminaries}\label{sec:preliminaries}
In this section, first, the power system of interest is modeled. Subsequently, the DUC model is formulated to introduce some decision variables and operational constraints of the general UC problem. After that, the RUC model in \cite{cho2019box} is explained, on which the exact WDRUC model is based. 

\subsection{Power System Modeling}
Throughout this paper, the UC problem for a power system with $I$ buses over a discrete planning horizon of $T$ time periods is considered, whose index sets are denoted by $\mathcal I$ and ${\mathcal T}$, respectively. Each bus has a generator, a load, and an REG system, all with the same index. The load at each bus in each time period is assumed to be initially known, while the REG is uncertain. It is also assumed that the load shedding and REG curtailment are possible. To model the transmission network, a DC power flow representation is used. Let $\mathcal L$ denote the index set of the transmission lines. 

The uncertainty of REG is addressed in the form of that of forecast error. Let $w^{\rm f}_{it}\in{\mathcal R}$ denote the REG forecast at bus $i$ in time period $t$, where $\mathcal R$ represents the real number set. Let also $w_{it}\in{\mathcal R}$ and ${\bf w}\in{\mathcal R}^{IT}$ denote the uncertain forecast error at bus $i$ in time period $t$ and the vector of $w_{it}$ for all $i\in{\mathcal I}$ and $t\in{\mathcal T}$, respectively. For each $i\in{\mathcal I}$ and $t\in{\mathcal T}$, $w_{it}$ takes a value in $\left[\underline{w}_{it},\overline{w}_{it}\right]\subset{\mathcal R}$ where $\underline{w}_{it}:=-w^{\rm f}_{it}$ and $\overline{w}_{it}:=W_i - w^{\rm f}_{it}$ with $W_i$ denoting the capacity of REG system $i$. Accordingly, the uncertainty set of $\bf w$ is defined as ${\mathcal W} := \left[\underline{\bf w},\overline{\bf w}\right]\subset{\mathcal R}^{IT}$, where $\underline{\bf w}$ and $\overline{\bf w}$ are the vectors of $\underline{w}_{it}$ and $\overline{w}_{it}$ for all $i\in{\mathcal I}$ and $t\in{\mathcal T}$, respectively. 
For each $t\in{\mathcal T}$, let ${\bf w}_t\in{\mathcal R}^{I}$ denote the vector of $w_{it}$ for all $i\in{\mathcal I}$.

\subsection{DUC}
Assuming that $w_{it}$ is known for each $i\in{\mathcal I}$ and $t\in{\mathcal T}$, the DUC model is formulated as follows: 
\begin{subequations}\label{eq:DUC}
\begin{align}
\begin{split}
&\min_{u^{\rm o}_{it},u^{\rm u}_{it},u^{\rm d}_{it},{x}^{\rm g}_{it},{x}^{\rm l}_{it},x^{\rm r}_{it}}
\sum_{i\in{\mathcal I}}\sum_{t\in{\mathcal T}}
\big(
C^{\rm o}_{i}u^{\rm o}_{it}+C^{\rm u}_{i}u^{\rm u}_{it}\\
&\quad +C^{\rm d}_{i}u^{\rm d}_{it}+C^{\rm g}_ix^{\rm g}_{it} + C^{\rm l}_{it} x^{\rm l}_{it} + C^{\rm r}_{it}x^{\rm r}_{it}\big)
\end{split}\label{eq:DUCobj}\\
&\text{s.t. \quad}u^{\rm o}_{it},u^{\rm u}_{it},u^{\rm d}_{it}\in\left\{0,1\right\},\quad
\forall i\in{\mathcal I},\forall t\in{\mathcal T},\label{eq:DUCcon1}\\
&
u^{\rm u}_{it}\geq u^{\rm o}_{it} - u^{\rm o}_{i(t-1)},\quad
\forall i\in{\mathcal I},\forall t\in{\mathcal T},\label{eq:DUCcon2}\\
&u^{\rm d}_{it}\geq -u^{\rm o}_{it} + u^{\rm o}_{i(t-1)},\quad
\forall i\in{\mathcal I},\forall t\in{\mathcal T},\label{eq:DUCcon3}\\
&u^{\rm o}_{it} + u^{\rm d}_{it} \leq 1,\quad
\forall i\in{\mathcal I},\forall t\in{\mathcal T},\label{eq:DUCcon4}\\
&u^{\rm o}_{i(t-1)} + u^{\rm u}_{it} \leq 1,\quad
\forall i\in{\mathcal I},\forall t\in{\mathcal T},\label{eq:DUCcon5}\\
\begin{split}
&u^{\rm o}_{it} - u^{\rm o}_{i(t-1)} \leq u^{\rm o}_{i\tau},\\
&\quad t \leq \tau \leq \min \left\{ t-1 + T^{\rm u}_i , T \right\}, \quad
\forall i\in{\mathcal I},\forall t\in{\mathcal T},
\end{split}\label{eq:DUCcon6}\\
\begin{split}
&u^{\rm o}_{i(t-1)} - u^{\rm o}_{it} \leq 1- u^{\rm o}_{i\tau}, \\
&\quad t \leq \tau \leq \min \left\{ t-1 + T^{\rm d}_i , T \right\},\quad
\forall i\in{\mathcal I},\forall t\in{\mathcal T},
\end{split}\label{eq:DUCcon7}\\
&x^{\rm g}_{it},x^{\rm r}_{it},x^{\rm l}_{it}\in{\mathcal R},
\quad\forall i\in{\mathcal I},\forall t\in{\mathcal T},\label{eq:DUCcon8}\\
&\underline{X}_i u^{\rm o}_{it} \leq {x}^{\rm g}_{it} \leq \overline{X}_iu^{\rm o}_{it},\quad
\forall i\in{\mathcal I},\forall t\in{\mathcal T},\label{eq:DUCcon9}\\
\begin{split}
&-X^{\rm rd}_i u^{\rm o}_{it} - X^{\rm sd}_i u^{\rm d}_{it}
\leq x^{\rm g}_{it} - x^{\rm g}_{i\left(t-1\right)} \\
&\quad \leq X^{\rm ru}_{i} u^{\rm o}_{i(t-1)} + X^{\rm su}_i u^{\rm u}_{it},\quad\forall i\in{\mathcal I},\forall t\in{\mathcal T},
\end{split}\label{eq:DUCcon10}\\
&0\leq x^{\rm l}_{it} \leq d_{it},\quad\forall i\in{\mathcal I},\forall t\in{\mathcal T},\label{eq:DUCcon11}\\
&0\leq x^{\rm r}_{it} \leq w^{\rm f}_{it} + {w}_{it},\quad\forall i\in{\mathcal I},\forall t\in{\mathcal T},\label{eq:DUCcon12}\\
\begin{split}
&-{F}_l\leq \sum_{i\in{\mathcal I}}F_{il} \left(x^{\rm g}_{it} - d_{it} + x^{\rm l}_{it} + w^{\rm f}_{it} + {w}_{it} - x^{\rm r}_{it}\right)\\
&\quad\leq{F}_{l},\quad\forall l\in{\mathcal L},\forall t\in{\mathcal T},
\end{split}\label{eq:DUCcon13}\\
&\sum_{i\in{\mathcal I}} \left(x^{\rm g}_{it} - d_{it} + x^{\rm l}_{it} + w^{\rm f}_{it} + {w}_{it} - x^{\rm r}_{it}\right)=0,\forall t\in{\mathcal T}.\label{eq:DUCcon14}
\end{align}
\end{subequations}

In (\ref{eq:DUC}), binary variables $u^{\rm o}_{it}$, $u^{\rm u}_{it}$, and $u^{\rm d}_{it}$ denote the on/off, start-up, and shut-down status of generator $i$ in time period $t$, respectively. Specifically, $u^{\rm o}_{it}$, $u^{\rm u}_{it}$, and $u^{\rm d}_{it}$ being equal to 1 indicate that generator $i$ is turned on, starts up, and shuts down in time period $t$, respectively. Real variables $x^{\rm g}_{it}$, $x^{\rm l}_{it}$, and $x^{\rm r}_{it}$ denote the conventional generation, load shedding, and REG curtailment at bus $i$ in time period $t$, respectively.

In (\ref{eq:DUCobj}), $C^{\rm o}_{i}$, $C^{\rm u}_{i}$, $C^{\rm d}_{i}$, and $C^{\rm g}_i$ denote the no-load, start-up, shut-down, and marginal costs of generator $i$, respectively; $C^{\rm l}_{it}$, and $C^{\rm r}_{it}$ represent the marginal costs of load shedding and REG curtailment at bus $i$ in time period $t$, respectively. 

Constraints (\ref{eq:DUCcon2})--(\ref{eq:DUCcon5}) are logical constraints for the on/off, start-up, and shut-down status of the generators; (\ref{eq:DUCcon6}) and (\ref{eq:DUCcon7}) are the minimum up and down time constraints for the generators, where $T^{\rm u}_i$ and $T^{\rm d}_i$ denote the minimum up and down times of generator $i$, respectively; (\ref{eq:DUCcon9}) imposes the upper and lower limits on the conventional generation at each time period, where $\underline{X}_i$ and $\overline{X}_i$ denote the minimum and maximum possible power outputs of generator $i$, respectively; (\ref{eq:DUCcon10}) is the ramping constraint of the generators, where $X^{\rm rd}_i$, $X^{\rm sd}_i$, $X^{\rm ru}_i$, and $X^{\rm su}_i$ denote the ramp-down, shut-down-ramp, ramp-up, and start-up-ramp limits of generator $i$, respectively; (\ref{eq:DUCcon11}) and (\ref{eq:DUCcon12}) put the upper limits on the load shedding and REG curtailment, respectively, where $d_{it}$ is the load at bus $i$ in time period $t$. The inequalities in (\ref{eq:DUCcon13}) are the transmission line capacity constraint, where $F_l$ and $F_{il}$ denote the maximum possible power flow in line $l$ and the power shift factor of bus $i$ and line $l$, respectively. Constraint (\ref{eq:DUCcon14}) represents the supply--demand balance condition. 

Let (\ref{eq:DUC}) be compactly rewritten as follows: 
\[
\min_{{\bf u}\in{\mathcal U},{\bf x}} {\bf c}_1^\top{\bf u} + {\bf c}^\top_2{\bf x}
\text{\quad s.t.\quad (\ref{eq:DUCcon8})--(\ref{eq:DUCcon14})},
\]
where ${\bf u}$ and $\bf x$ are the vectors of all the binary and real decision variables, respectively; the cost coefficient vectors ${\bf c}_1$ and ${\bf c}_2$ are defined accordingly. The feasible set of $\bf u$ is defined as ${\mathcal U}:=\left\{{\bf u}: \text{(\ref{eq:DUCcon1})--(\ref{eq:DUCcon7})}\right\}$. 

Problem (\ref{eq:DUC}) is a standard mixed-integer linear programming (MILP) problem that can be solved by many off-the-shelf solvers \cite{linderoth2010milp}. The goal of (\ref{eq:DUC}) is to minimize the total operating cost, i.e., the sum of the fixed operating cost ${\bf c}_1^\top{\bf u}$ and variable operating cost ${\bf c}^\top_2{\bf x}$, while satisfying all the operational constraints (\ref{eq:DUCcon1})--(\ref{eq:DUCcon14}). However, as ${\bf w}$ is uncertain, its solution may not guarantee that the constraints are actually met. The following subsection explains how the RUC model in \cite{cho2019box} deals with the uncertainty of $\bf w$. 

\subsection{RUC}
To address the uncertainty of $\bf w$ in (\ref{eq:DUC}), the RUC model in \cite{cho2019box} adopts the multi-stage decision-making framework assuming one day-ahead (DA) stage and $T$ real-time (RT) stages, the latter of which correspond to the $T$ time periods on the planning horizon. At the DA stage, the allowable dispatch range of each generator for each time period is determined simultaneously with the commitment status (by solving (\ref{eq:RUC}) below). This range is defined so that all the values in it are implementable regardless of the power outputs of that generator in the other time periods, by which the non-anticipativity constraint can be naturally met. 
Subsequently, at each RT stage, the actual REG forecast error is disclosed and the conventional generation, load shedding, and REG curtailment are optimized (by solving the associated subproblem in (\ref{eq:f}) below). In terms of the objective function, the RUC model minimizes the worst-case total operating cost over the uncertainty set of REG forecast error. 

Let $\overline{x}^{\rm g}_{it}$ and $\underline{x}^{\rm g}_{it}$ denote the maximum and minimum allowable power outputs of generator $i$ in time period $t$, respectively, which form the allowable dispatch range. Let also $\overline{\bf x}^{\rm g}$ and $\underline{\bf x}^{\rm g}$ denote the vectors of $\overline{x}^{\rm g}_{it}$ and $\underline{x}^{\rm g}_{it}$ for all $i\in{\mathcal I}$ and $t\in{\mathcal T}$, respectively. The RUC model is written as 
\begin{equation}\label{eq:RUC}
\min_{{\bf u}\in{\mathcal U},\left(\overline{\bf x}^{\rm g},\underline{\bf x}^{\rm g}\right)\in{\mathcal X}^{\rm r}\left({\bf u}\right)} 
\left\{
{\bf c}_1^\top{\bf u} + 
\max_{{\bf w}\in\mathcal W}
f\left(\overline{\bf x}^{\rm g},\underline{\bf x}^{\rm g},{\bf w}\right)
\right\}
\end{equation}
where
\begin{equation}\label{eq:f} 
f\left(\overline{\bf x}^{\rm g},\underline{\bf x}^{\rm g},{\bf w}\right) := 
\min_{{\bf x}\in{\mathcal X}\left(\overline{\bf x}^{\rm g},\underline{\bf x}^{\rm g},{\bf w}\right)} 
{\bf c}^\top_2{\bf x}.
\end{equation}
In (\ref{eq:RUC}), ${\mathcal X}^{\rm r}\left({\bf u}\right)$ is defined as follows: 
\[
\begin{aligned}
&{\mathcal X}^{\rm r}\left({\bf u}\right):=\{
\left(\overline{\bf x}^{\rm g},\underline{\bf x}^{\rm g}\right):\\
&\quad\underline{X}_i u^{\rm o}_{it} \leq\underline{x}^{\rm g}_{it}\leq \overline{x}^{\rm g}_{it} \leq \overline{X}_iu^{\rm o}_{it},\quad
\forall i\in{\mathcal I},\forall t\in{\mathcal T},\\
&\quad\overline{x}^{\rm g}_{it} - \underline{x}^{\rm g}_{i\left(t-1\right)} \leq X^{\rm ru}_{i} u^{\rm o}_{i(t-1)} + X^{\rm su}_i u^{\rm u}_{it},\quad\forall i\in{\mathcal I},\forall t\in{\mathcal T},\\
&\quad\overline{x}^{\rm g}_{i\left(t-1\right)} - \underline{x}^{\rm g}_{it}  \leq X^{\rm rd}_i u^{\rm o}_{it} + X^{\rm sd}_i u^{\rm d}_{it},\quad\forall i\in{\mathcal I},\forall t\in{\mathcal T}\big\}.
\end{aligned}
\]
In (\ref{eq:f}), ${\mathcal X}\left(\overline{\bf x}^{\rm g},\underline{\bf x}^{\rm g},{\bf w}\right)$ is defined as follows: 
\[
\begin{aligned}
&{\mathcal X}\left(\overline{\bf x}^{\rm g},\underline{\bf x}^{\rm g},{\bf w}\right):=\{
{\bf x}: \text{(\ref{eq:DUCcon8}), (\ref{eq:DUCcon11})--(\ref{eq:DUCcon14}),}\\
&\quad\underline{x}^{\rm g}_{it} \leq {x}^{\rm g}_{it} \leq \overline{x}^{\rm g}_{it},\quad
\forall i\in{\mathcal I},\forall t\in{\mathcal T}\},
\end{aligned}
\]
which is without any dynamic constraint. This implies that (\ref{eq:f}) can be decomposed into $T$ subproblems each of which is a single-period dispatch problem parameterized with the REG forecast error in only the corresponding time period. 

The operational constraints (\ref{eq:DUCcon9}) and (\ref{eq:DUCcon10}) of generators integrated in the DUC model (\ref{eq:DUC}) are considered in defining ${\mathcal X}^{\rm r}\left({\bf u}\right)$ so that solutions to (\ref{eq:f}) satisfy them. Thus, the feasibility of (\ref{eq:RUC}) ensures that all the operational constraints (\ref{eq:DUCcon1})--(\ref{eq:DUCcon14}) in (\ref{eq:DUC}) can be met adaptively and non-anticipatively.

Problem (\ref{eq:RUC}) can be solved by the column-and-constraint-generation (C\&CG) algorithm \cite{zeng2013solving}, while (\ref{eq:f}) is a standard linear programming (LP) problem. However, solutions to (\ref{eq:RUC}) might be overly conservative as it does not exploit statistical attributes of the REG forecast error. The exact WDRUC model explained in the following section makes more economical decisions by means of Wasserstein DRO. 

\section{Exact WDRUC}\label{sec:exactWDRUC}
In this section, the exact WDRUC model is described. Although ignored in the RUC model, $\bf w$ as a random vector has its associated true distribution, which is unknown but partially observable through samples. To utilize this probabilistic information while hedging against its incompleteness, the exact WDRUC model builds an empirical distribution and a Wasserstein ball around it, over which the worst-case expected total operating cost is minimized. 
Meanwhile, the model builds the Wasserstein ball on only a subset of the uncertainty set $\mathcal W$ to reduce conservativeness, simultaneously ensuring the operational feasibility for all ${\bf w}\in{\mathcal W}$.
The subset is defined using a novel technique, which has computational advantages compared to existing alternatives. Additional details are described later in this section.

Assuming that $S$ samples $\hat{\bf w}^{1},\ldots,\hat{\bf w}^{S}$ of ${\bf w}$ are given, the exact WDRUC model is written as follows:
\begin{subequations}\label{eq:DRUC}
\begin{align}
&\min_{\substack{{\bf u}\in{\mathcal U},\\
\left(\overline{\bf x}^{\rm g},\underline{\bf x}^{\rm g}\right)\in{\mathcal X}^{\rm r}\left({\bf u}\right)}}
\left\{{\bf c}^\top_1{\bf u}+\max_{{\mathbb P}\in{\mathcal P}_{\varepsilon}\left({\mathbb P}_{\rm e},{\Omega}\right)}
{\text E}_{\mathbb P}\left[f\left(\overline{\bf x}^{\rm g},\underline{\bf x}^{\rm g},{\bf w}\right)\right]
\right\}\label{eq:DRUCobj}\\
&\qquad\ \text{s.t.}\qquad \ {\mathcal X}\left(\overline{\bf x}^{\rm g},\underline{\bf x}^{\rm g},{\bf w}\right)\neq\emptyset,\quad\forall {\bf w}\in{\mathcal W}.\label{eq:DRUCcon}
\end{align}
\end{subequations}
In (\ref{eq:DRUCobj}), ${\mathbb P}_{\rm e}:=(1/S)\sum_{s\in{\mathcal S}}\delta_{\hat{\bf w}^{s}}
$ denotes the empirical distribution of ${\bf w}$,
where $\mathcal S$ and $\delta_{\bf w}$ represent the index set of the samples and the Dirac distribution with the unit mass at ${\bf w}$, respectively. 
Furthermore, $\Omega:={\mathcal W}\cap\left[\underline{\bf w}^{\rm a},\overline{\bf w}^{\rm a}\right]$ is a subset of $\mathcal W$, 
where $\underline{\bf w}^{\rm a}$ and $\overline{\bf w}^{\rm a}$ denote the vectors of $\underline{w}^{\rm a}_{it}$ and $\overline{w}^{\rm a}_{it}$ for all $i\in{\mathcal I}$ and $t\in{\mathcal T}$, respectively. The entries are defined as $\overline{w}^{\rm a}_{it} := \max_{s\in{\mathcal S}} \hat{w}^{s}_{it} + \varepsilon \max\left\{S,{\beta}\right\}$ and $\underline{w}^{\rm a}_{it} := \min_{s\in{\mathcal S}} \hat{w}^{s}_{it} - \varepsilon\max\left\{S,\beta\right\}$ where $\hat{w}^{s}_{it}$ and $\beta$ are the entry of $\hat{\bf w}^{s}$ corresponding to that $w_{it}$ of ${\bf w}$ and a user-defined parameter, respectively. Moreover, $\mathcal P_\varepsilon\left({\mathbb P}_{\rm e},{\Omega}\right)$ denotes the Wasserstein ball of radius $\varepsilon$ centered at ${\mathbb P}_{\rm e}$ on $\Omega$, 
i.e., ${\mathcal P}_{\varepsilon}\left({\mathbb P}_{\rm e},{\Omega}\right):= \left\{{\mathbb P}\in{\mathcal P}\left(\Omega\right): d_{\rm w}\left({\mathbb P},{\mathbb P}_{\rm e}\right)\leq\varepsilon\right\}$,
where ${\mathcal P}\left({\Omega}\right)$ and $d_{\rm w}\left({\mathbb P},{\mathbb P}_{\rm e}\right)$ denote 
the family of all probability distributions supported on ${\Omega}$ and the Wasserstein distance from ${\mathbb P}$ to ${\mathbb P}_{\rm e}$, respectively. By convention, $d_{\rm w}$ is defined as the 1-Wasserstein metric with the 1-norm, i.e.,
$d_{\rm w}\left({\mathbb P},{\mathbb P}^\prime\right):=\inf_{\pi\in\Pi\left({\mathbb P},{\mathbb P}^\prime\right)}\int_{{\mathcal W}\times{\mathcal W}}\onenorm{{\bf w} - {\bf w}^\prime}\pi\left(d{\bf w},d{\bf w}^\prime\right)$, where $\Pi\left({\mathbb P},{\mathbb P}^\prime\right)$ is the set of all joint distributions for ${\bf w}$ and ${\bf w}^\prime$ with marginals ${\mathbb P}$ and ${\mathbb P}^\prime$, respectively, and $\onenorm{\,\cdot\,}$ represents the 1-norm. The symbol ${\text E}_{\mathbb P}$ denotes the expectation operator with respect to a given distribution $\mathbb P$ of ${\bf w}$. Constraint (\ref{eq:DRUCcon}) ensures the operational feasibility for all ${\bf w}\in{\mathcal W}$. 

The exact WDRUC model (\ref{eq:DRUC}) uses the Wasserstein ball $\mathcal P_\varepsilon\left({\mathbb P}_{\rm e},{\Omega}\right)$ on $\Omega$, not $\mathcal P_\varepsilon\left({\mathbb P}_{\rm e},{\mathcal W}\right)$ on $\mathcal W$, as the latter may be conservative once the operational feasibility is ensured for all ${\bf w}\in{\mathcal W}$. Meanwhile, $\beta$ affects the worst-case confidence level of $\Omega$. The following theorem states the relationship:
\begin{theorem}\label{theorem1}
The worst-case probability of the realization of $\bf w$ being outside $\Omega$ over $\mathcal P_\varepsilon\left({\mathbb P}_{\rm e},{\mathcal W}\right)$ is bounded above by the reciprocal of $\max\left\{\beta,S\right\}$, i.e.,
$\sup_{{\mathbb P}\in{\mathcal P}_{\varepsilon}\left({\mathbb P}_{\rm e},{\mathcal W}\right)}{\mathbb P}\left[{\bf w}\notin\Omega\right]\leq 1/\max\left\{\beta,S\right\}$.
\end{theorem}

\begin{proof}
Assume for the proof that $\bf w$ has a larger uncertainty set ${\mathcal W}^{\rm u}\subset{\mathcal R}^{IT}$, which is an arbitrary compact convex set with an interior containing ${\mathcal W}\cup\left[\underline{\bf w}^{\rm a},\overline{\bf w}^{\rm a}\right]$. According to Theorem 4.4 and Corollary 5.3 in \cite{esfahani2018data}, $
V := \sup_{{\mathbb P}\in{\mathcal P}_\varepsilon\left({\mathbb P}_{\rm e},{\mathcal W}^{\rm u}\right)}{\mathbb P}\left[{\bf w}\notin\left(\underline{\bf w}^{\rm a},\overline{\bf w}^{\rm a}\right)\right]
$
is equal to the optimal value of the following problem: 
\begin{equation}\label{eq:theorembeta}
\begin{aligned}
&\sup_{\alpha^{s}_k\geq0,{\bf q}^{s}_k}
&&
\frac{1}{S}\sum_{s\in{\mathcal S}}\sum_{k\in{\mathcal K}}\alpha^{s}_{k}l_{k}\left(\hat{\bf w}^s + {\bf q}^{s}_k/{\alpha^s_k}\right)\\
&\text{s.t.}&&\frac{1}{S}\sum_{s\in{\mathcal S}}\sum_{k\in{\mathcal K}}\onenorm{{\bf q}^s_k}\leq \varepsilon,\\
&&&\sum_{k\in{\mathcal K}}\alpha^s_{k} = 1,\quad\forall s\in{\mathcal S},\\
&&&\hat{\bf w}^s + 
{\bf q}^{s}/{\alpha^s_k} \in {\mathcal W}^{\rm u},\quad\forall k\in{\mathcal K},\forall s\in{\mathcal S}
\end{aligned}
\end{equation}
where ${\mathcal K}:=\left\{1,2,\ldots,2IT+1\right\}$, 
\[
\begin{aligned}
&l_k\left({\bf w}\right):=
\begin{cases}
\begin{aligned}
&1 && \text{if}\quad {\bf e}_k^\top{\bf w} \geq {\bf e}_k^\top\overline{\bf w}^{\rm a}\\
&-\infty && \text{otherwise}
\end{aligned}
\end{cases}, 1 \leq k \leq IT,\\
&l_{IT+k}\left({\bf w}\right):=
\begin{cases}
\begin{aligned}
&1 && \text{if}\quad {\bf e}_{k}^\top{\bf w} \leq {\bf e}_{k}^\top\underline{\bf w}^{\rm a} \\
&-\infty && \text{otherwise}
\end{aligned}
\end{cases}, 1 \leq k \leq IT,
\end{aligned}
\]
and $l_{2IT+1}\left({\bf w}\right):=0$ with ${\bf e}_k$ denoting the $k$th column of the $IT$-dimensional identity matrix. In (\ref{eq:theorembeta}), the conventional extended arithmetics apply, e.g., $1/0 = \infty$, $0/0 = 0$, and $0\cdot\infty = 0$. For each $\left(i,k,s,t\right)\in{\mathcal I}\times{\mathcal K}\times{\mathcal S}\times{\mathcal T}$, let $q^{s}_{kit}$ denote the entry of ${\bf q}^{s}_k$ corresponding to that $w_{it}$ of ${\bf w}$. The optimal value of (\ref{eq:theorembeta}) is obtained if either ${q}^s_{kit}=\varepsilon S$ for any $\left(i,k,s,t\right)\in{\mathcal I}\times\left\{1,\ldots,IT\right\}\times{\mathcal S}\times{\mathcal T}$ such that $s=\argmaxA_{s^\prime\in{\mathcal S}} \hat{w}^{s^\prime}_{it}$ and ${\bf e}_{k}^\top\hat{\bf w}^s = \hat{w}^s_{it}$, or ${q}^s_{kit}=-\varepsilon S$ for any $\left(i,k,s,t\right)\in{\mathcal I}\times\left\{IT+1,\ldots,2IT\right\}\times{\mathcal S}\times{\mathcal T}$ such that $s=\argminA_{s^\prime\in{\mathcal S}} \hat{w}^{s^\prime}_{it}$ and ${\bf e}_{k-IT}^\top\hat{\bf w}^{s} = \hat{w}^{s}_{it}$, with, in either case, $\alpha^{s}_k$ equal to $1$, if $S\geq\beta$, and $S/\beta$, otherwise. These cases are when a data point originally closest to the boundary of $\left[\underline{\bf w}^{\rm a},\overline{\bf w}^{\rm a}\right]$ moves along the shortest path to reach it, i.e., by $\varepsilon\max\left\{S,\beta\right\}$. That is to say, $V= 1/\max\left\{S,\beta\right\}$. Furthermore, it is observed that
\[
\begin{aligned}
V 
&\geq \sup_{{\mathbb P}\in{\mathcal P}_\varepsilon\left({\mathbb P}_{\rm e},{\mathcal W}^{\rm u}\right)}{\mathbb P}\left[{\bf w}\notin
\left[\underline{\bf w}^{\rm a},\overline{\bf w}^{\rm a}\right]\right]\\
&\geq \sup_{{\mathbb P}\in{\mathcal P}_\varepsilon\left({\mathbb P}_{\rm e},{\mathcal W}^{\rm u}\right)\cap
\left\{{\mathbb P}^\prime\in{\mathcal P}\left({\mathcal W}^{\rm u}\right): {\mathbb P}^\prime\left({\bf w}\in{\mathcal W}\right)=1\right\}
}{\mathbb P}\left[{\bf w}\notin\left[\underline{\bf w}^{\rm a},\overline{\bf w}^{\rm a}\right]\right]\\
&= \sup_{{\mathbb P}\in{\mathcal P}_\varepsilon\left({\mathbb P}_{\rm e},{\mathcal W}^{\rm u}\right)\cap\left\{{\mathbb P}^\prime\in{\mathcal P}\left({\mathcal W}^{\rm u}\right): {\mathbb P}^\prime\left({\bf w}\in{\mathcal W}\right)=1\right\}}{\mathbb P}\left[{\bf w}\notin{\Omega}\right]\\
&= \sup_{{\mathbb P}\in{\mathcal P}_\varepsilon\left({\mathbb P}_{\rm e},{\mathcal W}\right)}{\mathbb P}\left[{\bf w}\notin{\Omega}\right].
\end{aligned}
\]
Hence the statement holds. 
\end{proof}

Theorem \ref{theorem1} implies that $\left(1-1/\beta\right)$ can be interpreted as the worst-case confidence level of $\Omega$ for a given $\varepsilon$ when $\beta\geq S$. Thus, $\beta$ can be set based on the system operator's preference, e.g., to $1/0.05$ and $1/0.01$, similar to the confidence level of a chance-constrained optimization problem. 

The above-described way of defining the subset $\Omega$ over which the Wasserstein ball is built has not been studied. Although existing approaches, mostly as an approximation to Wasserstein distributionally robust chance constraints, e.g., \cite{duan2018distributionally} and \cite{poolla2020wasserstein}, might be used, the new technique is adopted as it is computationally less costly. For example, the method in \cite{duan2018distributionally} would define the subset as a convex hull of some uncertain scenarios, which is hard to represent as a system of constraints. Further, the resulting subset may not preserve the sample size independence in the tractable reformulation of the proposed model explained in the subsequent section. The method in \cite{poolla2020wasserstein} requires multiple optimization problems scaling with sample size to be solved. In contrast, no such potentially large-scale optimization problem is faced when constructing $\Omega$ in the aforementioned way. 

Problem (\ref{eq:DRUC}) can be rewritten as a two-stage RO problem for which many existing algorithms are applicable in theory \cite{esfahani2018data,zhao2018data}. For example, the following proposition holds. 
\begin{prop}\label{prop}
Problem (\ref{eq:DRUC}) is rewritten as
\begin{equation}\label{eq:exactreformulation}
\begin{aligned}
&\min_{\substack{
{\bf u}\in{\mathcal U},\left(\overline{\bf x}^{\rm g},\underline{\bf x}^{\rm g}\right)\in{\mathcal X}^{\rm e}\left({\bf u}\right),
\lambda\geq0,\eta^{s}}}
{\bf c}^\top_1{\bf u}+\lambda\varepsilon+\frac{1}{S}\sum_{s\in{\mathcal S}}\eta^{s}\\
&\text{s.t.\ }f\left(\overline{\bf x}^{\rm g},\underline{\bf x}^{\rm g},{\bf w}\right)-\lambda\onenorm{{\bf w}-\hat{\bf w}^{s}}\leq\eta^{s},\forall s\in{\mathcal S},\forall{\bf w}\in{\Omega}
\end{aligned}
\end{equation}
where ${\mathcal X}^{\rm e}\left({\bf u}\right):= \left\{\left(\overline{\bf x}^{\rm g},\underline{\bf x}^{\rm g}\right)\in{\mathcal X}^{\rm r}\left({\bf u}\right):\text{(\ref{eq:DRUCcon})}\right\}$. 
\end{prop}

It can be shown using the well-known WDRO techniques in \cite{gao2016distributionally}.\footnote{The proof is provided in Appendix I of the extended version \cite{cho2022affine}.} Unfortunately, the resulting problem may still be hard to tackle in practice, partly because its scale increases with the sample size S, as can be seen from (\ref{eq:exactreformulation}). In the following section, the proposed model as an approximation to (\ref{eq:DRUC}) is presented, which has a tractable exact reformulation with a scale invariant with $S$.

Meanwhile, the quality of a solution to (\ref{eq:DRUC}) depends on $\varepsilon$. For a general class of Wasserstein DRO problems including  (\ref{eq:DRUC}), it is shown in the literature that a careful tuning of $\varepsilon$ leads to some desirable performance guarantees \cite{esfahani2018data}. However, it is also reported that a Wasserstein ball constructed in such a way might be overly conservative. To obtain a less conservative solution, $\varepsilon$ can be set in a data-driven manner, e.g., by the holdout method or cross-validation \cite{esfahani2018data}.

\section{WDRUC with Affine Policy}\label{sec:WDRUCwith}
In this section, first, the proposed WDRUC model is formulated. Subsequently, its reformulation with a scale independent of sample size is provided as a solution method. 

\subsection{WDRUC with Affine Policy}
The proposed model applies an affine policy to the exact WDRUC model (\ref{eq:DRUC}), by which each decision variable of (\ref{eq:f}) is modeled as an affine function of the total REG forecast error in the corresponding time period \cite{lorca2016multistage}. Specifically, the conventional generation, load shedding, and REG curtailment at bus $i$ in time period $t$ are defined by the affine policy approach as 
$x^{\rm ga}_{it}({\bf w}_t) := a^{\rm g1}_{it} \sum_{i^\prime\in{\mathcal I}}w_{i^\prime t} + a^{\rm g0}_{it}$, $x^{\rm la}_{it}({\bf w}_t) := a^{\rm l1}_{it}\sum_{i^\prime\in{\mathcal I}}w_{i^\prime t} + a^{\rm l0}_{it}$, and $x^{\rm ra}_{it}({\bf w}_t) := a^{\rm r1}_{it}\sum_{i^\prime\in{\mathcal I}}w_{i^\prime t}
+ a^{\rm r0}_{it}$, respectively. Here, real numbers $a^{\rm g1}_{it}$, $a^{\rm g0}_{it}$, $a^{\rm l1}_{it}$, $a^{\rm l0}_{it}$, $a^{\rm r1}_{it}$, and $a^{\rm r0}_{it}$ are determined with $\bf u$, $\overline{\bf x}^{\rm g}$, and $\underline{\bf x}^{\rm g}$. Let ${\bf a}^1$ and ${\bf a}^0$ denote the vector of $a^{\rm g1}_{it}$, $a^{\rm l1}_{it}$, and $a^{\rm r1}_{it}$, and that of $a^{\rm g0}_{it}$, $a^{\rm l0}_{it}$, and $a^{\rm r0}_{it}$ for all $i\in{\mathcal I}$ and $t\in{\mathcal T}$, respectively. Let also ${\bf a}^1_t$ and ${\bf a}^0_t$ for each $t\in{\mathcal T}$ denote the vector of $a^{\rm g1}_{it}$, $a^{\rm l1}_{it}$, and $a^{\rm r1}_{it}$, and that of $a^{\rm g0}_{it}$, $a^{\rm l0}_{it}$, and $a^{\rm r0}_{it}$ for all $i\in{\mathcal I}$, respectively. 

As the worst-case expected total operating cost is minimized over ${\mathcal P}_{\varepsilon}\left({\mathbb P}_{\rm e},{\Omega}\right)$ in (\ref{eq:DRUC}), the affine policy is applied only for $\Omega$. The proposed model is formulated as follows: 
\begin{subequations}\label{eq:ADRUC}
\begin{align}
&\min_{\substack{{\bf u}\in{\mathcal U},\left(\overline{\bf x}^{\rm g},\underline{\bf x}^{\rm g}\right)\in{\mathcal X}^{\rm e}\left({\bf u}\right),
{\bf a}^1,{\bf a}^0}}
{\bf c}^\top_1{\bf u}+g\left({\bf a}^1,{\bf a}^0\right)\label{eq:ADRUCobj}\\
&\text{s.t.}\quad\left({\bf a}^1,{\bf a}^0\right)\in{\mathcal A}\left(\overline{\bf x}^{\rm g},\underline{\bf x}^{\rm g},{\bf w}\right),\forall{\bf w}\in{\Omega}.\label{eq:ADRUCcon1}
\end{align}
\end{subequations}
In (\ref{eq:ADRUCobj}), $g\left({\bf a}^1,{\bf a}^0\right)$ denotes the worst-case expected operating cost over ${\mathcal P}_{\varepsilon}\left({\mathbb P}_{\rm e},{\Omega}\right)$ incurred by the affine policy, i.e., 
\begin{equation}\label{eq:g}
\begin{aligned}
&g\left({\bf a}^1,{\bf a}^0\right):=\\
&\quad\max_{{\mathbb P}\in{\mathcal P}_{\varepsilon}\left({\mathbb P}_{\rm e},{\Omega}\right)}{\text E}_{\mathbb P}
\Big[\sum_{t\in{\mathcal T}} \Big( c^1_t \left({\bf a}^1_t\right)\sum_{i\in{\mathcal I}}w_{it} + c^0_t\left({\bf a}^0_t\Big)\right)\Big]
\end{aligned}
\end{equation}
where $c^1_t$ and $c^0_t$ for each $t\in{\mathcal T}$ are linear functions whose coefficients are defined using the entries of ${\bf c}_2$. In (\ref{eq:ADRUCcon1}), ${\mathcal A}\left(\overline{\bf x}^{\rm g},\underline{\bf x}^{\rm g},{\bf w}\right)$ is obtained by substituting the affine functions for ${\bf x}$ in the constraints of ${\mathcal X}\left(\overline{\bf x}^{\rm g},\underline{\bf x}^{\rm g},{\bf w}\right)$. 
Meanwhile, the affine policy is discarded in the actual dispatch because of its suboptimality. In the actual dispatch, the subproblems of (\ref{eq:f}) are solved, which are always feasible due to (\ref{eq:DRUCcon}).

Similar to the exact WDRUC model (\ref{eq:DRUC}), the proposed model (\ref{eq:ADRUC}) is not solvable in the current form. The tractable reformulation of (\ref{eq:ADRUC}) is given in the following subsection.

\subsection{Tractable Reformulation}
As the most important contribution of this paper, this subsection provides the tractable exact reformulation of (\ref{eq:ADRUC}). The reformulation is based on the following theorem:
\begin{theorem}\label{theorem2}
For each $t\in{\mathcal T}$, let $\overline{z}^+_t := \sum_{i\in{\mathcal I}}\sum_{s\in{\mathcal S}} \overline{v}^s_{it}$ and $\overline{z}^-_t := -\sum_{i\in{\mathcal I}}\sum_{s\in{\mathcal S}} \underline{v}^s_{it}$ where $\underline{v}^s_{it}:=\max\left\{\underline{w}_{it},\underline{w}^{\rm a}_{it}\right\} - \hat{w}^s_{it}$ and $\overline{v}^s_{it}:=\min\left\{\overline{w}_{it},\overline{w}^{\rm a}_{it}\right\} - \hat{w}^{s}_{it}$. It holds that 
\[
g\left({\bf a}^1,{\bf a}^0\right) = g^{\rm c}\left({\bf a}^1,{\bf a}^0\right)+g^{\rm v}\left({\bf a}^1\right)\]
where
\[
g^{\rm c}\left({\bf a}^1,{\bf a}^0\right): =\sum_{t\in{\mathcal T}}
\Big(c^{1}_t\left({\bf a}^1_t\right) \sum_{i\in{\mathcal I}}\sum_{s\in{\mathcal S}}\frac{1}{S}\hat{w}^{s}_{it} + c^{0}_t\left({\bf a}^0_t\right)\Big) 
\]
and $g^{\rm v}\left({\bf a}^1\right)$ denotes the optimal value of the LP problem
\begin{equation}\label{eq:h}
\begin{aligned}
&\max_{{z^+_{t},z^-_{t}}\geq0}&&\frac{1}{S}\sum_{t\in{\mathcal T}}c^{1}_t\left({\bf a}^1_t\right)\left(z^+_{t} - z^-_{t}\right)\\
&\text{s.t.}&&\frac{1}{S}\sum_{t\in{\mathcal T}}
\left(z^+_{t} + z^-_{t}\right)\leq \varepsilon,\\
&& &z^+_{t}\leq\overline{z}^+_t,\quad z^-_{t} \leq\overline{z}^-_{t},\quad\forall t\in{\mathcal T}.
\end{aligned}
\end{equation}
\end{theorem}
\begin{proof}
According to Theorem 4.4 in \cite{esfahani2018data}, (\ref{eq:g}) is rewritten as 
\[
\begin{aligned}
&\max_{v^{s}_{it}}
&&
\frac{1}{S}\sum_{s\in{\mathcal S}}\sum_{t\in{\mathcal T}}\Big(c^{1}_t\left({\bf a}^1_t\right)\sum_{i\in{\mathcal I}}\left(\hat{w}^{s}_{it} + v^{s}_{it}\right)+c^{0}_t\left({\bf a}^0_t\right)\Big)\\
&\text{s.t}&&
\frac{1}{S}\sum_{i\in{\mathcal I}}\sum_{s\in{\mathcal S}}\sum_{t\in{\mathcal T}}
\onenorm{v^{s}_{it}} \leq \varepsilon,\\
&&&\underline{v}^{\rm s}_{it} \leq v^{s}_{it} \leq \overline{v}^s_{it},\quad \forall i\in{\mathcal I},\forall s\in{\mathcal S},\forall t\in{\mathcal T}.
\end{aligned}
\]
Thus, it is enough to show that the following problem is equivalent to (\ref{eq:h}): 
\begin{equation}\label{eq:h2}
\begin{aligned}
&\max_{v^{s}_{it}}&&
\frac{1}{S}\sum_{i\in{\mathcal I}}\sum_{s\in{\mathcal S}}\sum_{t\in{\mathcal T}}c^{1}_t\left({\bf a}^1_t\right)v^{s}_{it}\\
&\text{s.t}&&
\frac{1}{S}\sum_{i\in{\mathcal I}}\sum_{s\in{\mathcal S}}\sum_{t\in{\mathcal T}}\onenorm{v^{s}_{it}} \leq \varepsilon,\\
&&&\underline{v}^{\rm s}_{it}\leq v^{s}_{it} \leq \overline{v}^s_{it},\quad \forall i\in{\mathcal I},\forall s\in{\mathcal S},\forall t\in{\mathcal T}.
\end{aligned}
\end{equation}
By introducing auxiliary variables to linearize the norm constraint, (\ref{eq:h2}) is rewritten as follows: 
\begin{equation}\label{eq:prototype2}
\begin{aligned}
&\max_{v^{s+}_{it},v^{s-}_{it}\geq0}&&
\frac{1}{S}\sum_{i\in{\mathcal I}}\sum_{s\in{\mathcal S}}\sum_{t\in{\mathcal T}}c^{1}_t\left({\bf a}^1_t\right)\left(v^{s+}_{it} - v^{s-}_{it}\right)\\
&\text{s.t.}&&
\frac{1}{S}\sum_{i\in{\mathcal I}}\sum_{s\in{\mathcal S}}\sum_{t\in{\mathcal T}}\left(v^{s+}_{it}+v^{s-}_{it}\right) \leq \varepsilon,\\
&&&v^{s+}_{it} \leq \overline{v}^{s}_{it},\quad\forall i\in{\mathcal I},\forall s\in{\mathcal S},\forall t\in{\mathcal T},\\
&&&v^{s-}_{it} \leq -\underline{v}^{s}_{it},\quad \forall i\in{\mathcal I},\forall s\in{\mathcal S},\forall t\in{\mathcal T}.
\end{aligned}
\end{equation}
In (\ref{eq:prototype2}), $v^{s+}_{it}$ for all $i\in{\mathcal I}$ and $s\in{\mathcal S}$ have an identical coefficient in the objective function and the first constraint, which is the only coupling constraint; the same holds for $v^{s-}_{it}$. By letting $z^+_t:=\sum_{i\in{\mathcal I}}\sum_{s\in{\mathcal S}} v^{s+}_{it}$ and $z^-_t:=\sum_{i\in{\mathcal I}}\sum_{s\in{\mathcal S}} v^{s-}_{it}$ for each $t\in{\mathcal T}$, (\ref{eq:prototype2}) is rewritten as (\ref{eq:h}) without loss of optimality.
\end{proof}

Based on duality in LP, (\ref{eq:h}) is equivalent to 
\begin{equation}\label{eq:hdual}
\begin{aligned}
&\min_{\xi,\xi^+_t,\xi^-_t\geq0}&&\varepsilon\xi + \sum_{t\in{\mathcal T}}\left(\overline{z}^+_{t}\xi^+_t + \overline{z}^-_{t}\xi^-_t\right)\\
&\text{s.t.}&&\xi + S\xi^+_t \geq c^{1}_{t}\left({\bf a}^1_t\right),\quad\forall t\in{\mathcal T},\\
&& &\xi + S\xi^-_t \geq-c^{1}_t\left({\bf a}^1_t\right),\quad \forall t\in{\mathcal T}
\end{aligned}
\end{equation}
where $\xi$, $\xi^+_t$, and $\xi^-_t$ are the dual variables associated with the first constraint, the upper limit constraint on $z^+_{t}$, and the upper limit constraint on $z^-_{t}$ in (\ref{eq:h}), respectively. Let (\ref{eq:hdual}) be compactly rewritten as 
\[
\min_{\bm\upxi\in{\Xi}({\bf a}^1)} {\bf c}^\top_3{\bm\upxi}.
\]
Thus, the proposed model (\ref{eq:ADRUC}) is expressed as 
\begin{equation}\label{eq:ADRUCdual}
\begin{aligned}
&\min_{\substack{{\bf u}\in{\mathcal U},\left(\overline{\bf x}^{\rm g},\underline{\bf x}^{\rm g}\right)\in{\mathcal X}^{\rm e}\left({\bf u}\right),\\{\bm\upxi}\in\Xi({\bf a}^1),{\bf a}^1,{\bf a}^0,}}&&{\bf c}^\top_1{\bf u}+g^{\rm c}\left({\bf a}^1,{\bf a}^0\right)+{\bf c}^\top_3{\bm\upxi}\\
&\text{s.t.}&&\text{(\ref{eq:ADRUCcon1})}.
\end{aligned}
\end{equation}

Problem (\ref{eq:ADRUCdual}) is a semi-infinite programming problem with constraints (\ref{eq:DRUCcon}) and (\ref{eq:ADRUCcon1}), without which it is a standard MILP problem. These two constraints can be addressed by well-studied RO techniques. For example, the C\&CG algorithm in \cite{zeng2013solving} and the cutting-plane algorithm in \cite{lorca2016multistage} can deal with (\ref{eq:DRUCcon}) and (\ref{eq:ADRUCcon1}), respectively. Further details of the algorithm for (\ref{eq:ADRUCdual}) based on the two existing algorithms are explained in Appendix II of the extended version \cite{cho2022affine}. 

The most distinctive feature of (\ref{eq:ADRUCdual}) is that its scale remains invariant with $S$, thereby, unlike in \cite{zheng2020data} and \cite{gamboa2021decomposition}, allowing a number of samples to be easily exploited in the Wasserstein DRO framework without using computationally intensive decomposition algorithms. Although the WDRUC model in \cite{zhu2019wasserstein} has a similar property, it ignores the non-anticipativity constraint by design and thus is impractical. On the other hand, in the Wasserstein distributionally robust control method \cite{yang2020wasserstein}, the non-anticipativity constraint is naturally satisfied but the invariance in computational complexity with respect to sample size is nontrivial to achieve in general.\footnote{A notable exception is the linear-quadratic control case \cite{kim2021distributional}, in which the computational complexity of the Riccati equation is independent of the sample size.}

The decision-making process of the proposed model (\ref{eq:ADRUC}) is similar to that of (\ref{eq:DRUC}) and (\ref{eq:RUC}). Specifically,  (\ref{eq:ADRUCdual}) is solved at the DA stage and the $T$ subproblems of (\ref{eq:f}) are sequentially solved at the $T$ RT stages. Fig. \ref{fig:DM2} summarizes the DA decisions of the three models. 

 \begin{figure}[t!]
    \centering
 \vspace{2mm}
       \includegraphics[width=8.5cm]{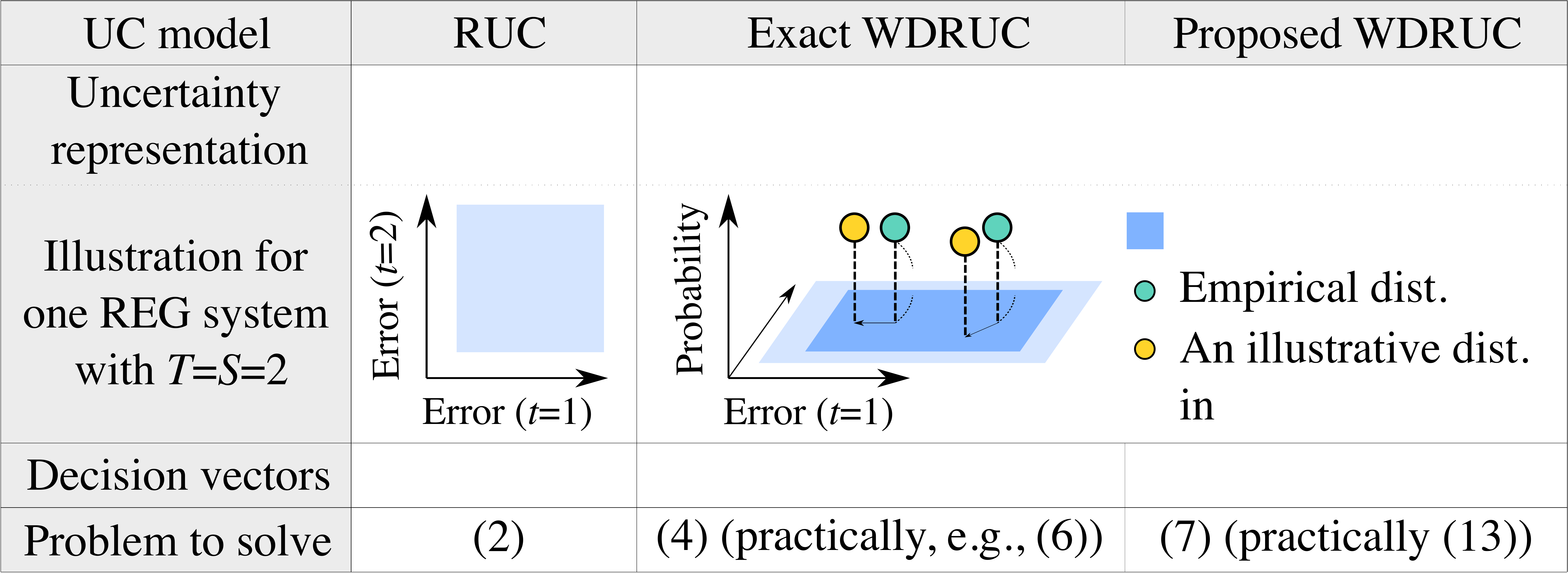}
               \put(-52.5,23.5){\footnotesize 
${\mathcal P}_{\varepsilon}\left({\mathbb P}_{\rm e},{\Omega}\right)$}
               \put(-16.5,41.7){\footnotesize 
${\mathbb P}_{\rm e}$}
               \put(-59.5,50.5){\footnotesize 
$\Omega$}
               \put(-169.5,66.5){\footnotesize 
$\mathcal W$}
               \put(-169.5,49){\footnotesize 
$\mathcal W$}
               \put(-98,66.5){\footnotesize 
$\mathcal W$, $\Omega$, ${\mathcal P}_{\varepsilon}\left({\mathbb P}_{\rm e},{\Omega}\right)$}
\put(-180.5,12.5){\footnotesize ${\bf u},\overline{\bf x}^{\rm g},\underline{\bf x}^{\rm g}$}
\put(-119.5,12.5){\footnotesize ${\bf u},\overline{\bf x}^{\rm g},\underline{\bf x}^{\rm g}$}
\put(-62.5,12.5){\footnotesize ${\bf u},\overline{\bf x}^{\rm g},\underline{\bf x}^{\rm g},{\bf a}^{\rm 1},{\bf a}^{\rm 0}$}
\put(-101.5,45){\tiny $\frac{1}{2}$}
\put(-85.5,45){\tiny $\frac{1}{2}$}
\put(-108.5,35){\tiny $\varepsilon$}
\put(-90.5,34.5){\tiny $\varepsilon$}
\vspace{-2mm}
    \caption{Day-ahead decisions of the RUC, the exact WDRUC, and the proposed models.}
    \vspace{-4mm}
\label{fig:DM2}
\end{figure}

\section{Numerical Simulations}\label{sec:simulations}
In this section, the proposed model (A-WDRUC) is compared to a UC model based on SP (SUC), the RUC model (\ref{eq:RUC}) (RUC), and the exact WDRUC model (\ref{eq:DRUC}) (E-WDRUC) on 6- and 24-bus test systems. Here, SUC is formulated as 
\[
\min_{{\bf u}\in{\mathcal U},\left(\overline{\bf x}^{\rm g},\underline{\bf x}^{\rm g}\right)\in{\mathcal X}^{\rm r}\left({\bf u}\right)}
{\bf c}^\top_1{\bf u} + \frac{1}{S}\sum_{s\in{\mathcal S}}f\left(\overline{\bf x}^{\rm g},\underline{\bf x}^{\rm g},\hat{\bf w}^s\right).
\]
The generator, load, and branch data of the 6- and the 24-bus systems are from \cite{IIT} and \cite{ordoudis2016updated}, respectively. The 6- and the 24-bus systems have one and three photovoltaic (PV) systems, respectively. The planning horizon consists of 24 one-hour time periods. For simplicity, only up to 10 preselected loads are assumed sheddable, while all the PV systems are curtailable. The marginal costs of load shedding and PV curtailment are set to the largest marginal generation cost multiplied by 100 and zero, respectively. The radius of the Wasserstein ball for A-WDRUC is chosen from 
$\left\{b\times10^{-h}:b=1,5, h=1,2,3\right\}$ by the holdout method. SUC and RUC are solved without any decomposition method to address a large number of samples and with the C\&CG algorithm, respectively. E-WDRUC is solved in the form of (\ref{eq:exactreformulation}) by the C\&CG algorithm. The simulations were run on MATLAB with CPLEX 12.10. The source code of our A-WDRUC implementation is available online.\footnote{\tt\small https://github.com/CORE-SNU/A-WDRUC}

\subsection{Comparison to SUC and RUC}
In this subsection, A-WDRUC is compared to SUC and RUC in terms of the computation time and the expected total operating cost evaluated with respect to the true distribution of the PV forecast error (or simply ``cost") for different sample sizes. First, for a given PV forecast \cite{nrel}, the true distribution of the forecast error for each PV system and time period is modeled by the normal distribution with a mean of zero and a standard deviation of 0.2 times the forecast, according to which $S$ samples are randomly generated. Subsequently, the UC models are solved. Then, the expected total operating costs are calculated for a uniform distribution of another 10,000 scenarios randomly generated according to the true distribution, which approximates the cost. For statistical robustness, this process is repeated 50 times. 

The results for the 6-bus system are illustrated in Fig. \ref{fig:sim1res_6bus}, which are averaged over all the simulation runs. Fig. \ref{fig:sim1res_6bus}a shows that the computation time of A-WDRUC does not increase with the sample size, while that of SUC does rapidly. Fig. \ref{fig:sim1res_6bus}b and \ref{fig:sim1res_6bus}c further reveal that both A-WDRUC and SUC tend to yield lower costs with more samples. However, the performance of A-WDRUC is less sensitive to changes in sample size compared to that of SUC. Although this indicates that SUC can outperform A-WDRUC when a large number of samples are available, the drastically increasing computational load hinders its usage. Meanwhile, A-WDRUC consumes more computation time but achieves lower costs on average than RUC for all $S$ except $S=2$. 
 \begin{figure}[t!]
    \centering
       \includegraphics[width=8.5cm]{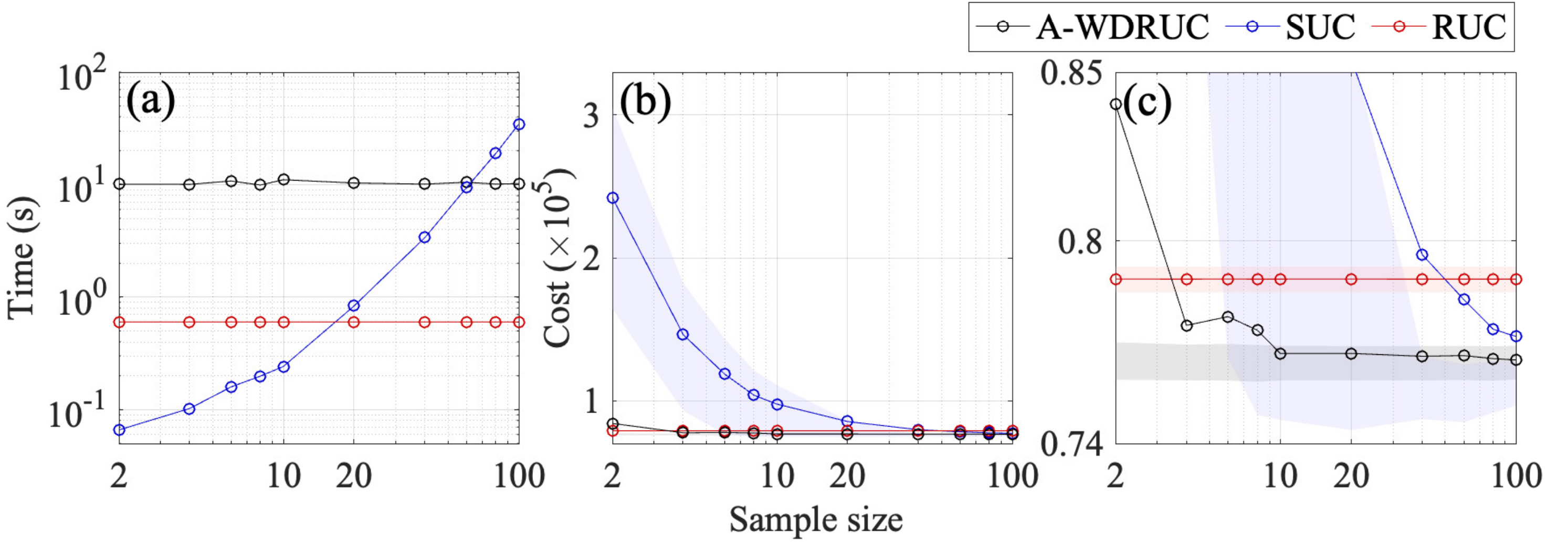}
       \vspace{-2mm}
    \caption{
    Comparison of A-WDRUC, SUC, and RUC on 
    the 6-bus system: (a) computation time, (b) cost, and (c) magnification of (b). In (b) and (c), the shaded areas represent ranges between the 25th and 75th percentiles.}
        \vspace{-1mm}
\label{fig:sim1res_6bus}
\end{figure}
Fig. \ref{fig:sim1res_24bus} depicts the averaged results for the 24-bus system, which are similar to those for the 6-bus system except that SUC is unavailable for $S\geq80$ due to memory shortage. 

 \begin{figure}[t!]
    \centering
       \includegraphics[width=8.5cm]{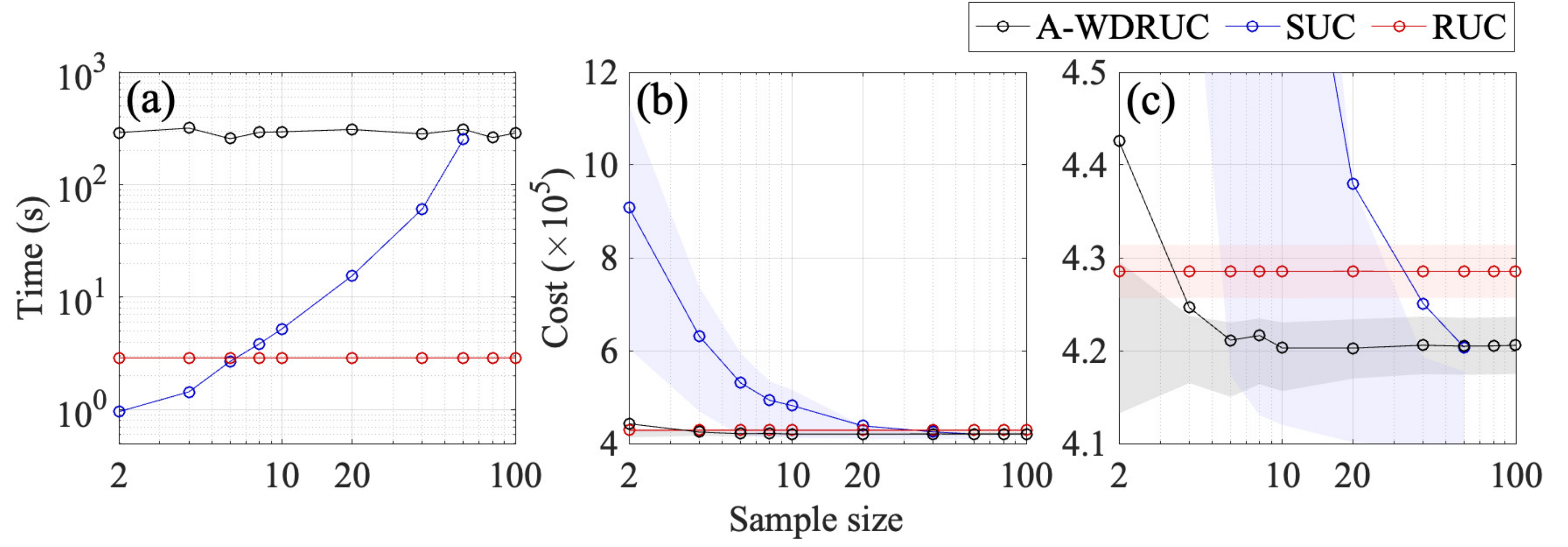}
        \vspace{-2mm}
    \caption{Comparison of A-WDRUC, SUC, and RUC on the 24-bus system: (a) computation time, (b) cost, and (c) magnification of (b).}
\label{fig:sim1res_24bus}
\end{figure}

Overall, the simulation results demonstrate that the computation time of A-WDRUC does not depend on the sample size and that A-WDRUC, on average, outperforms both SUC and RUC in terms of the cost when only a moderate number of samples are at hand. 

\subsection{Comparison to E-WDRUC}
In this subsection, A-WDRUC is further compared to E-WDRUC. Considering the anticipated huge computational load of E-WDRUC, only 10 simulation runs were conducted on the 6-bus system for each $S\leq40$. 

The averaged results are plotted in Fig. \ref{fig:sim1res_EWDRUC}, from which it can be noted that both computation time and cost of E-WDRUC vary sensitively with $S$, similar to those of SUC. However, differently from SUC, E-WDRUC is computationally more intensive than A-WDRUC even for $S=2$. Surprisingly, E-WDRUC yields notably higher costs than A-WDRUC for the smallest values of $S$, which is due to the regularization effect of the affine policy. This implies that A-WDRUC is less vulnerable to the optimizer's curse than E-WDRUC \cite{smith2006optimizer}. In summary, the simulation results indicate that A-WDRUC is more useful than E-WDRUC for the same reason that the former is more useful than SUC.
\begin{figure}[t!]
    \centering
       \includegraphics[width=8.5cm]{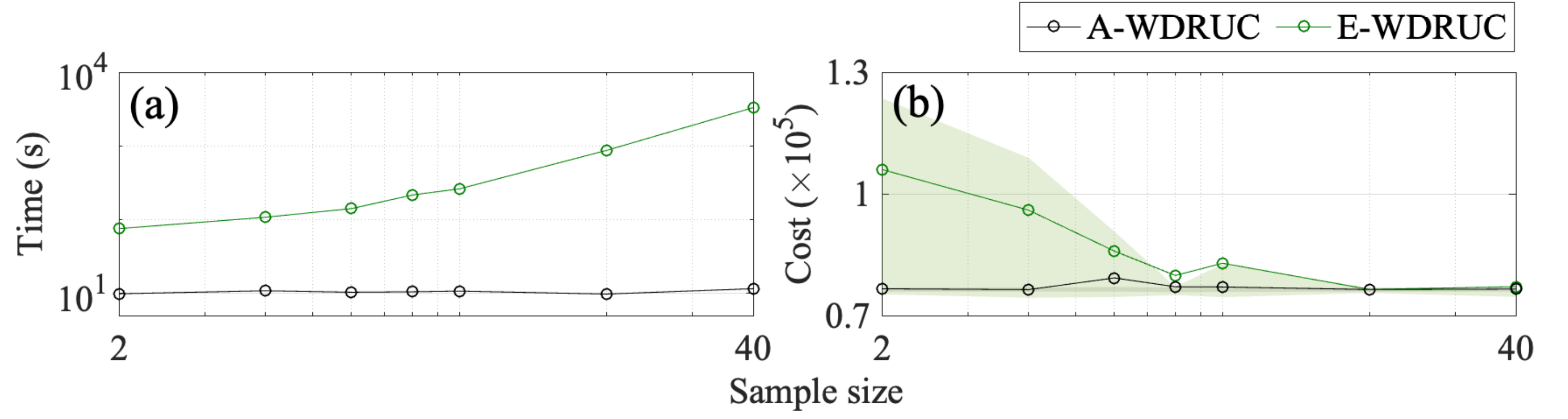}
       \vspace{-2mm}
    \caption{Comparison of A-WDRUC and E-WDRUC on the 6-bus system: (a) computation time and (b) cost.}
        \vspace{-4mm}
\label{fig:sim1res_EWDRUC}
\end{figure}

\section{Conclusions}\label{sec:conclusions}
In this paper, we have proposed a WDRUC model and its tractable exact reformulation. First, the proposed model was formulated as a WDRO problem, relying on an affine policy that satisfies the non-anticipativity constraint. Subsequently, the reformulation was given, with a scale invariant to the sample size. As a result, a number of samples can be easily exploited. Moreover, a novel technique was used to build a subset of the uncertainty set for the Wasserstein ball, further endowing the proposed model with economic efficiency. The numerical simulations verified the computational and economic effectiveness of the proposed model. 

\bibliographystyle{IEEEtran}
\bibliography{references}

\end{document}